\documentclass[1p,12pt]{Myelsarticle}

\usepackage{amssymb,amsmath}
\usepackage{graphicx}  
\usepackage{graphics}
\usepackage{wrapfig}
\usepackage{multirow}
\usepackage{subfigure}

\newtheorem{theorem}{Theorem}

\newtheorem{corollary}{Corollary}
\newtheorem{claim}{Claim}

\newproof{proof}{Proof}

\newcommand{\inlineproofenvironment}[1][Proof]{\noindent\textsc{#1.}}

\begin{document}

\begin{frontmatter}

\title{Complexity of Canadian Traveler Problem Variants}

\author[bgu]{Dror Fried\corref{cor1}}
\ead{friedd@cs.bgu.ac.il}
\author[bgu]{Solomon Eyal Shimony}
\ead{shimony@cs.bgu.ac.il}
\author[bgu]{Amit Benbassat}
\ead{amitbenb@cs.bgu.ac.il}
\author[kth]{Cenny Wenner}
\ead{cenny@cwenner.net}

\cortext[cor1]{Corresponding author}
\address[bgu]{Department of Computer Science, Ben-Gurion University of the Negev
P.O.B. 653 Beer-Sheva 84105 Israel
}
\address[kth]{KTH - Royal Institute of Technology and Stockholm University, 
NADA, Att: Cenny Wenner
SE-10044, Stockholm}

\begin{abstract}

The Canadian traveler problem (CTP) is the problem of traversing a given graph,
where some of the edges may be blocked -- a state which is revealed only upon
reaching an incident vertex.
Originally stated by Papadimitriou and Yannakakis (1991), 
the adversarial version of CTP was shown to be PSPACE-complete, with
the stochastic version shown to be \#P-hard.

We show that stochastic CTP is also PSPACE-complete:
initially proving PSPACE-hardness for the dependent version of stochastic CTP,
and proceeding with gadgets that
allow us to extend the proof to the independent case.

Since for disjoint-path graphs, CTP can be solved in polynomial time,
we examine the complexity of the more general remote-sensing CTP,
and show that it is NP-hard even for disjoint-path graphs.

\end{abstract}


\end{frontmatter}
\section{Introduction}\label{sec:intro}

In the {\em stochastic Canadian traveler problem} (CTP) \cite{PY} we are
given an undirected connected weighted graph $G=(V,E)$,
a source vertex ($s \in V$), and a target vertex ($t \in V$).
Any edge $e\in E$ may be {\em blocked} with a known probability $p(e)$.
The actual state of each edge $e\in E$ becomes known
only upon reaching a vertex incident on $e$. Traversing an unblocked edge $e$ incurs
a non-negative cost equal to the weight of $e$.
The problem is to find a policy $\pi$ that
minimizes the expected traversal cost $C(\pi)$ from $s$ to $t$.

CTP formalizes a basic question of navigating in a partially known environment,
which is a fundamental task for transportation, autonomous robotic systems, computer games, and more.
Other variants of CTP have been introduced and analyzed in the research literature \cite{BNS,EK,PolyTsi90}.
There has been a strong recent resurgence of interest in CTP,
both theoretical \cite{XHSZZ, SW} and empirical \cite{PTM, RepCTP, BFS}.
A preliminary alternative proof
of Theorem \ref{thm:CTPtheorem} appears in an unpublished 
work by one of the authors \cite{wenner09hardness}.

When originally introduced in \cite{PY}, two variants were examined:
the adversarial variant and the stochastic variant. The adversarial variant was shown to be 
PSPACE-complete by reduction from QSAT. For the stochastic version, membership in PSPACE was shown, however only
\#P-hardness was established by reduction from the \emph{st-reliability problem}, leaving the question
of PSPACE-hardness open. Apparently proving the stronger result requires some form 
of dependency between the edges, achieved ``through the back door'' in
the adversarial variant. This paper settles the question, showing that
CTP is indeed PSPACE-complete.

Since the size of an optimal policy is potentially exponential in the size of the problem
description, we in fact show that it is PSPACE-hard to find even the optimal first action
at $s$. 

We begin with a variant
of CTP with dependent directed edges, \emph{CTP-Dep}, which allows for a simple proof of PSPACE-hardness
by reduction from QSAT, before proceeding with the proof 
for the ``standard'' stochastic CTP. Although the latter
result subsumes the former, proving the dependent CTP result first greatly simplifies
the intuition behind the proof of the standard case.

Another variant we explore is remote-sensing CTP, henceforth called \textit{Sensing-CTP}, in which additional
actions called \textit{remote-sensing actions} are allowed. 
Each such action reveals, for a certain cost, the status of a non-incident edge.
Recently it was shown \cite{BFS} that stochastic CTP can be solved in low-order
polynomial time on disjoint-path graphs. It was believed that generalizing
CTP to allow remote-sensing actions makes the problem harder -- indeed we show
that allowing remote-sensing makes CTP NP-hard even on disjoint-path graphs.

\section{Dependent directed CTP is PSPACE-hard}\label{sec:CTP-Dep}

This general form of \textit{dependent CTP} (called CTP-Dep) is a 5-tuple $(G, w, s, t, B)$ 
with $G=(V,E)$ a directed graph, a weight function $w:E\to\Re^{\geq 0}$, 
$s, t \in V$ are the start and goal vertices respectively, and
a distribution model $B$ over binary random variables indexed by the edges $E$.
We assume that $B$ is specified as a Bayes network  
over these random variables $E$ \cite{Pearl} as follows. Each random variable expresses
the state (blocked, unblocked) of an edge in $E$ (abusing notation we use the symbols indicating
the edges to denote the respective random variables). 
The Bayes network $(E, A, P)$ consists of a set of directed arcs $A$ between the random variables $E$, so that $(E,A)$ is a directed acyclic
graph. $P$ describes the conditional probability tables, one for each $e\in E$.

\begin{theorem}\label{thm:CTPDep}
CTP-Dep is PSPACE-hard.
\end{theorem}

\inlineproofenvironment\ by reduction from QSAT \cite{GJ}.
Recall that QSAT is the language of all satisfiable quantified boolean formulas (QBF), 
$\Phi = \forall x_1\exists x_2 ... \varphi(x_1,x_2,..., x_n)$, where $\varphi$ is a boolean formula in conjunctive
normal form, with $n$ variables and $m$ clauses, which contain literals, each is consisting of either a variable or a negated variable.
We assume that each clause has at most 3
literals. 
Given a QBF $\Phi$, construct a CTP-Dep instance $(G_\Phi, w, s, t, B)$ as follows
(see Fig.\ \ref{fig:QBFreduction}). $G_\Phi$ consists of a \textit{variables section},
and an \textit{exam section}. Vertices in the variables section have labels 
starting with $v$ or $o$, and vertices of the exam section begin with $r$.
An always unblocked edge $(s,t)$, called the \textit{default edge}, has a cost of $h$.
All other edges, unless mentioned otherwise, are zero-cost edges known to be
unblocked.

\begin{figure}[h]
    \centering
    \includegraphics[scale=1]{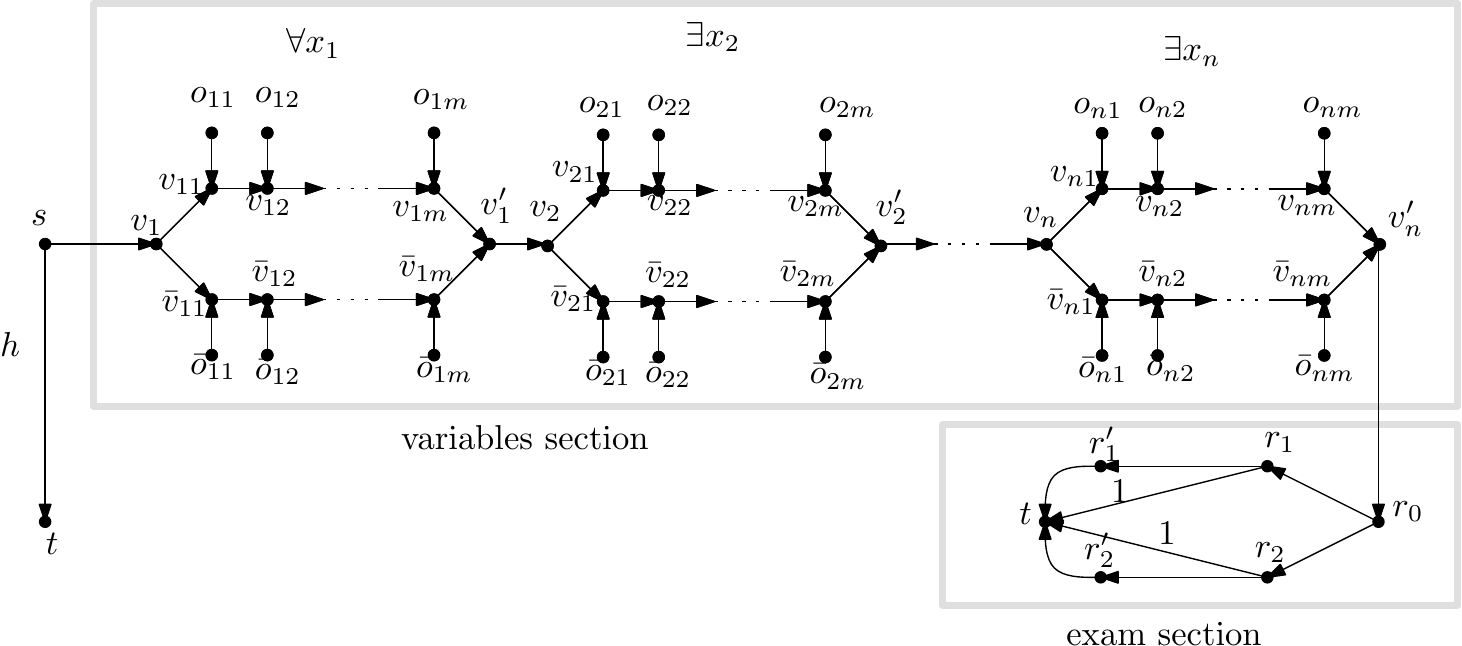}
    \caption{Reduction from QBF to CTP-Dep. Note that vertex $t$ appears twice in order to simplify the physical layout.}
    \label{fig:QBFreduction}
\end{figure}
The variables section contains a subsection $X_i$ for every variable $x_i$,
which begins at $v_i$, and ends at $v'_i$. 
For every $i<n$, $X_i$ is connected to $X_{i+1}$ 
through an edge $(v'_i,v_{i+1})$. 

Every $X_i$ contains a 
\textit{true-path} $(v_i,v_{i1},\cdots,v_{im},v'_i)$, and a 
\textit{false-path} $(v_i,\bar{v}_{i1},\cdots,\bar{v}_{im},v'_i)$. 
If $x_i$ is a universal variable (resp. existential variable), the edges $(v_i,v_{i1})$, and $(v_i,\bar{v}_{i1})$ 
are called \textit{universal edges} (resp. \textit{existential edges}). While the existential edges are
always unblocked, we set the universal
edges to have have blocking probability $1/2$ and to be mutually
exclusive: for each universal variable $x_i$, exactly one of $(v_i,v_{i1})$, 
and $(v_i,\bar{v}_{i1})$ is blocked. 

In addition, for every $1\leq i\leq n$, and $1\leq l\leq m$, 
there are edges $(o_{il},v_{il})$, and $(\bar{o}_{il},\bar{v}_{il})$,
called \textit{observation edges}. 
These edges
are only meant to be observed, as their source vertices are unreachable. 
Every observation edge is blocked with probability $1/2$, and the dependency of the observation edges
is defined according to appearance of
variables in the clauses of $\Phi$, as follows: an observation edge $(o_{il},v_{il})$ 
(resp. $(\bar{o}_{il},\bar{v}_{il})$) is 
considered ``in'' a clause $C_l$ if $x_i$ appears unnegated (resp. negated) in clause $C_l$.
All observation edges that are ``in'' the same clause $C_l$ co-occur: they are
either all traversable or all are blocked (with probability $1/2$, as stated above),
independent of all other edges that are not ``in'' $C_l$.

The exam section consists of an \textit{odd-path} $(r_0,r_1,r_1',t)$, 
and an \textit{even-path} $(r_0,r_2,r_2',t)$. In addition construct
edges $(r_1,t)$, and $(r_2,t)$ with cost $1$. 
The edges $(r_1,r'_1)$, and $(r_2,r'_2)$ are called \textit{choice edges}. 
The edge $(r_1, r'_1)$ (resp. $(r_2, r'_2)$) is unblocked if and only if the observation edges are
unblocked for an {\em odd} (resp. {\em even}) number of clauses. Hence exactly one of the choice edges is blocked
\footnote{Note that as every clause has at most three literals, this dependency
structure can be realized with a Bayes network
of constant in-degree, a construction that has polynomial size.}.
If at least one observation edge in each clause is observed, the status 
of the choice edges can be determined with certainty.
Otherwise the posterior blocking probability of each choice edge remains $1/2$. 
In order to prove the theorem, it is sufficient to prove the following claim:

\begin{claim}
An optimal policy has expected cost $0$ just when $\Phi$ is satisfiable
(in which case the optimal first action is to traverse $(s,v_1)$).
Otherwise (for any $h<2^{-\frac{n}{2}-1}$)  the optimal policy is to traverse
 $(s,t)$ with a cost of $h$.
\end{claim}

\begin{proof}
Suppose first that $\Phi$ is satisfiable.
Then there is a policy for
assigning values to all the existential variables, each given every setting of the
enclosing universal variables, such that $\varphi$ is true. Following this policy for
each existential variable $x_i$, i.e. traversing edge $(v_i, v_{i1})$ if
$x_i$ should be \textit{true}, and $(v_i, \bar{v}_{i1})$ otherwise,
leads (by construction) to following a path such that at least one observation edge
is seen in every clause. Hence, the ``exam'' is passed (i.e. the 0-cost unblocked path in the exam section is chosen)
with certainty.

Next, suppose $\Phi$ is not satisfiable. Then there is at least one setting of
the universal variables for which some clause $C_l$ is false under the same conditions,
and thus no edge ``in'' clause $C$ is observed.
Since every setting of these variables occurs with probability $2^{-\frac{n}{2}}$
(assuming w.l.o.g. that $n$ is even), in these cases the exam is ``flunked'' (picking the 
path where only the expensive edge is unblocked)
with probability $1/2$, and thus the total expected cost of
starting with $(s, v_1)$ is at least $2^{-\frac{n}{2}-1}$. 
Hence, setting $h <2^{-\frac{n}{2}-1}$, the optimal policy is
to traverse $(s, t)$ if and only if $\Phi$ is not satisfiable. \qed

\end{proof}

\section{Complexity of CTP}\label{sec:ComCTP}

Having shown that CTP-Dep is PSPACE-hard, we extend the proof to the ``standard'' stochastic
independent undirected edges CTP:

\begin{theorem}\label{thm:CTPtheorem}
CTP is PSPACE-complete.
\end{theorem}

In order to prove Theorem \ref{thm:CTPtheorem}, we use the same general outline of the reduction from QBF as in the proof of Theorem \ref{thm:CTPDep}.
However, in CTP-Dep, dependencies and directed edges restrict the available choices, thereby simplifying
the proof. Here we introduce special gadgets that limit choice de facto,
and show that any deviation from these limitations is necessarily sub-optimal.
Policies that obey these limitation are called {\em reasonable policies}.
Each such gadget $g$ has an {\em entry} terminal $Entry(g)$, and an {\em exit} terminal $Exit(g)$;
an attempt to traverse $g$ from $Entry(g)$ to $Exit(g)$ is henceforth called to {\em cross $g$}.
The gadgets operate by allowing a potential {\em shortcut} to the target $t$; crossing these gadgets
may either end up at $Exit(g)$, with some probability $q(g)$, or at $t$ instead. 
The edges that allow direct access to $t$ are called \textit{shortcut} edges.

We introduce the gadgets in sections \ref{sec:BG} and \ref{sec:OG}, and
the CTP-graph construction in Section \ref{sec:CTP-cons}. The actual proof of Theorem \ref{thm:CTPtheorem} is
in Section \ref{sec:ThmProof}.
In the description of the gadgets and CTP-graph, we sometimes add zero cost always traversable edges. These edges, 
which appear unlabeled in figures \ref {fig:BGsubfig1},\ref{fig:OGsubfig2} and \ref{fig:CTPReduction}, were added solely in order to simplify the physical layout as a figure;
any $u$, $v$ connected by such an edge can be considered to be the same vertex.

\subsection{Baiting Gadgets}\label{sec:BG}

A \textit{baiting gadget} $g=BG(u,v)$ with parameter $L>1$ is a three-terminal weighted graph
(see Fig.\ \ref{fig:BGsubfig1}): 
an entry terminal $u=Entry(g)$, an exit terminal $v=Exit(g)$,
and a shortcut terminal which is always $t$.
The latter terminal is henceforth omitted in external reference to $g$, for conciseness.

The baiting gadget consists of $N+1$ uniform sections of an undirected path $(u,v_1,\cdots,v_N,v)$ with
total weight $L$,
each intermediate vertex has a 0-cost shortcut to $t$ with a blocking probability $1/2$.
In addition, there is a shortcut edge with cost $L$ from the terminals $u,v$ to $t$.
Set $N=2^{\left\lceil \log_2(4L)\right\rceil}-1$.
We assume that $g$ is connected to the graph such that any policy executed at $u$, in which
the edge $(u,v_1)$ is not traversed, has an expected cost of at least $1$. Later on we see that
this assumption holds in the CTP-graph construction.

Let $\pi$ be the following partial policy: \textit{when at $u$ for the first 
time, proceed along the path $(u,v_1,\cdots,v_N,v)$ to $v$, taking the 0-cost shortcut to $t$ whenever possible, 
but never backtracking to $u$.}

It is easy to show that even if we need to take the cost $L$ shortcut at $v$, the
expected cost of executing $\pi$ at $u$ for the first time is less than $1$. 
Because of the $L$ cost shortcut edge $(v,t)$, the expected cost of any optimal policy once at $v$
(knowing all 0-cost shortcuts are blocked) is no more than $L$, hence under 
any reasonable policy, $g$ is not retraced. A similar argument holds for retracing to $u$ from 
other locations along the path
$(u,v_1,\cdots,v_N,v)$. Hence we have:

\begin{figure}[ht]
    \centering
    \includegraphics[scale=1.5]{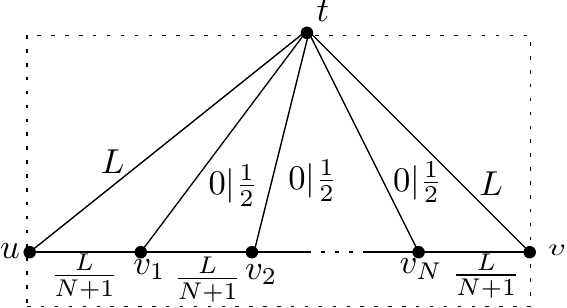}
    \caption{A Baiting Gadget $BG(u,v)$ with a parameter $L>1$. Edge label $w|p$ denotes cost$|$blocking probability.
    The optimal policy at $u$ is to cross the path $(u,v_1,\cdots,v_N,v)$, taking a shortcut edge to $t$
    whenever such an edge is found unblocked. After reaching $v$, retracing to $u$ in $g$ costs at least $L$. }
    \label{fig:BGsubfig1}
\end{figure}

\begin{claim}\label{clm:BGcLaim}
When at $u$ for the first time, $\pi$ is optimal for a baiting gadget $BG(u,v)$ with a parameter $L>1$.
After reaching $v$, it is suboptimal to backtrack to $u$ in $g$.
\end{claim}

Note that $g$ is actually symmetric w.r.t. $u,v$. However, since by construction of the CTP-graph,
every reasonable policy always reaches one designated terminal $u$
first, we treat $g$ externally as if it were directional. A precise derivation of
the parameters of baiting gadgets appears in \ref{app:ApA1}.

\subsection{Observation Gadgets}\label{sec:OG}

An \textit{observation gadget} $g=OG(u,v,o)$ is a four-terminal weighted graph (see Fig. \ref{fig:OGsubfig2}): 
an entry terminal $u=Entry(g)$, an exit terminal $v=Exit(g)$, an observation terminal $o$,
and a shortcut terminal (again omitted in external
references) which is always $t$.
The observation gadget begins with a baiting gadget $BG_1=BG(u,v_1)$ with parameter $L>8$, which is connected
to the ``observation loop'' beginning with a baiting gadget $BG_2=BG(v_1, v_2)$ 
with parameter $3L/2$, a zero-cost edge
$(v_2,v_3)$ with blocking probability $3/4$, and a cost $L_1=5L/8$ unblocked edge to $o$. 
An always unblocked $3L/2$ shortcut edge $(v_2,t)$ is assumed.
The observation loop is closed by a cost $L_1$ unblocked edge to $v_4$ and a zero-cost
edge $(v_4, v_1)$ with blocking probability $3/4$. From $v_1$, a cost $1$ unblocked edge  $(v_1,v')$ followed
by a baiting gadget $BG_3=BG(v'_1, u)$ with parameter $L$ completes the gadget.

\begin{figure}[ht]
    \centering
    \includegraphics[scale=1.3]{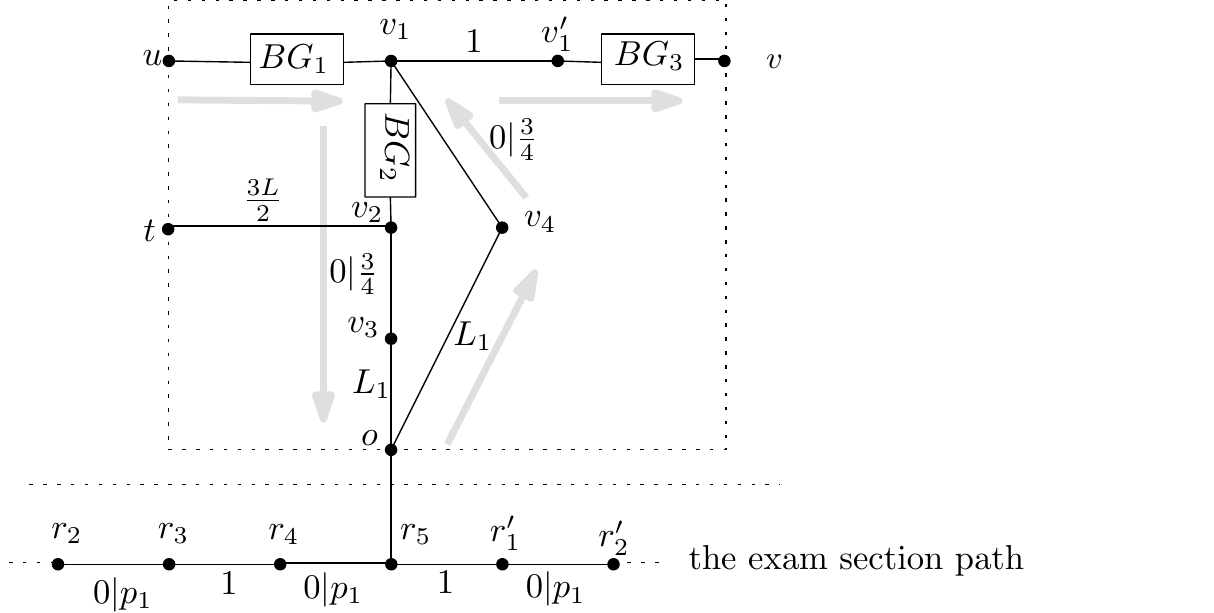}
    \caption{An Observation Gadget $OG(u,v,o)$. Light gray arrows indicate general traversal direction of the optimal policy $\pi$.}
    \label{fig:OGsubfig2}
\end{figure}

We assume that $o$ is either not directly connected to the rest of the graph,
or connected through a path $(r_2,r_3,r_4,r_5,r'_1,r'_2)$ called the \textit{exam section path} ($o$ is identified with $r_5$) with the following properties: 
the edges $(r_2,r_3)$, $(r'_1,r'_2)$ and $(r_4,r_5)$, have zero cost and blocking probability $p_1$, where $p_1>1-2/(3L+1)$. 
$(r_2,r_3)$ and  $(r'_1,r'_2)$ are called \textit{guard edges}, $(r_4,r_5)$ is called an \textit{observation edge}. 
The edges $(r_3,r_4)$ and $(r_5,r'_1)$ are always traversable edges with cost 1. 
The vertex $o$ is allowed to coincide with observation terminals of other observation gadgets.
The notations of the exam section path are chosen to match the description of the CTP-graph construction in Section \ref{sec:CTP-cons}.

Let $\pi_g$ be the following partial policy at $g$: \textit{when at $u$, cross $BG_1$.
Then (observing $(v_1,v_4)$), cross $BG_2$. If either $(v_2,v_3)$, or $(v_1,v_4)$ is found blocked,
reach $t$ by traversing the cost $3L/2$ shortcut edge $(v_2,t)$.
However, if both $(v_2,v_3)$ and $(v_1,v_4)$ are unblocked, 
traverse $(v_2,v_3,o,v_4,v_1,v'_1)$ (observing any edges incident on $o$ such as the observation edge $(r_4,r_5)$),
and cross $BG_3$. }

Again, by construction of the CTP-graph (section \ref{sec:CTP-cons}),
any policy at $u$ other than crossing $BG_1$ results in a cost of at least $1$.

\begin{claim}\label{clm:OGclaim}
When at $u$ for the first time, $\pi_g$ is an optimal policy for $g$.
\end{claim}

\inlineproofenvironment[Proof Outline]\ Properties of the baiting gadgets ensure that $g$ is traversed in the
correct order. The guard edges $(r_2,r_3)$ and $(r'_1,r'_2)$ ensure that it is suboptimal to ``escape'' from $o$ by traversing
edges in the exam section.
The uncertain edges $(v_4,v_1)$ and $(v_2, v_3)$ ensure that it is suboptimal to
enter a previously uncrossed observation gadget from $o$. 
Likewise for a previously crossed observation
gadget $g'$: entering $g'$ through $o$ is suboptimal because all the baiting gadgets in $g'$
have been crossed and observed to
contain no unblocked zero-cost shortcuts. 

\vspace{2mm}

Detailed derivation of the properties of observation gadgets appears in \ref{app:ApA2}.

\subsection{CTP-graph Construction}\label{sec:CTP-cons}

\begin{figure}[!h]
 
   \scalebox{1}{\includegraphics[width=1.0\textwidth,height=0.5\textheight]{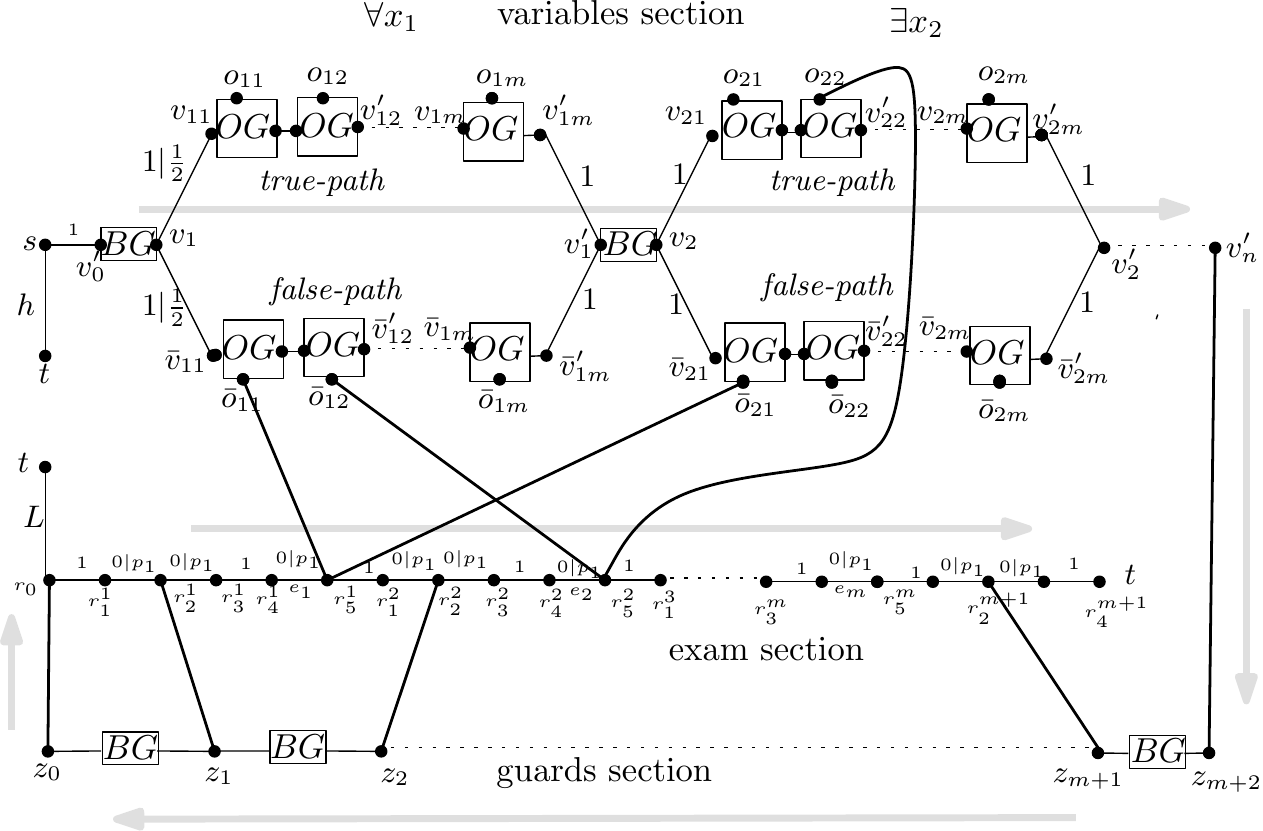}}
    \caption{CTP-graph construction for $\Phi = \forall x_1\exists x_2\cdots(\bar{x}_1\vee\bar{x}_2)\wedge (\bar{x}_1\vee x_2)\cdots$. $BG$ - a baiting gadget. $OG$ - an observation gadget. Light gray arrows indicate the general traversal direction of the optimal policy when $\Phi$ is satisfiable.}
    \label{fig:CTPReduction}
\end{figure}

Having shown the properties of the baiting and observation gadgets, we are ready to construct the CTP-graph:
For a QBF $\Phi$ with $n$ variables and $m$ clauses, we construct 
$G_\Phi$ in the same general outline as the construction of the CTP-Dep graph (see section \ref{sec:CTP-Dep}) 
with the following changes (see Fig.\ \ref{fig:CTPReduction}). 
The exam section is a path of $5(m+1)+1$ vertices $\{r^i_j|1\leq i\leq m+1,1\leq j\leq 5\}$
as follows. 
For every $0<i\leq m+1$, $(r^i_1,r^i_2)$, $(r^i_2,r^i_3)$ and $(r^i_4,r^i_5)$ have zero
cost and blocking probability $p_1$, apart from $(r^{m+1}_4,r^{m+1}_5)$ which has zero cost and is always traversable. 
$(r^i_1,r^i_2)$, and $(r^i_2,r^i_3)$ are called \textit{guard edges},
and $(r^i_4,r^i_5)$ is called a \textit{clause edge}.
$(r^i_3,r^i_4)$ and  $(r^i_5,r^{i+1}_1)$ are always traversable cost $1$ edges. In addition, the exam section holds an additional vertex $r_0$, an always traversable cost $1$ edge $(r_0,r^1_1)$, and an always traversable cost $L$ edge $(r_0,t)$.
In order to guarantee correct operation of the observation gadgets,
we disallow reasonable policies to traverse exam edges too early while crossing the variable section. This is done
by visiting the initially uncertain guard edges only later via a section called the \textit{guards section}, which consists
of a sequence of baiting gadgets $BG(z_i,z_{i+1})$ with parameter $L$ that
visit $r^i_2$ for all $0\leq i \leq m+1$. 

The variables section is constructed as for CTP-Dep, 
except that the directed edges $(v'_i,v_{i+1})$ are replaced by
baiting gadgets $BG(v'_i,v_{i+1})$ with parameter $L$. 
For each universal variable $x_i$ the 
universal edges $(v_i,v_{i1})$, and $(v_i,\bar{v}_{i1})$ are cost $1$ edges with 
blocking probability $1/2$. For each existential variable $x_i$,
the existential edges $(v_i,v_{i1})$, and $(v_i,\bar{v}_{i1})$ are always traversable edges with cost $1$. 
Inside each \textit{true-path}, every $(v_{ij},o_{ij})$, 
$(v_{ij},v_{i(j+1)})$ pair is replaced by an observation gadget $g=OG(v_{ij},v'_{ij},o_{ij})$. $(v'_{ij},v_{i(j+1)})$ 
are always unblocked zero-cost edges added for clarity.
The observation vertex $o_{ij}$ is identified with
the vertex incident on the appropriate clause edge in
the exam section. 
That is, if $x_i$ appears unnegated in clause $j$, then
$o_{i_j}$ of the \textit{true-path} is identified with $r^j_5$ in the exam section.
Likewise respectively for all the edges in the \textit{false-paths}.

For example, Fig.\ \ref{fig:CTPReduction} demonstrates the reduction for
$\Phi = \forall x_1\exists x_2\cdots(\bar{x}_1\vee\bar{x}_2)\wedge (\bar{x}_1\vee x_2)\cdots$.
The variable $x_1$ appears negated in clause $2$, so in $G_\Phi$ the vertex $\bar{o}_{12}$ at the section $X_1$, and the vertex $r_5^2$ of the exam section are connected by an unlabeled edge, hence the clause edge $e_2=(r_4^2,r_5^2)$ can be observed from the observation gadget $OG(\bar{v}_{12},\bar{v}'_{12},\bar{o}_{12})$ when traversing the \textit{false} path of $X_1$. Likewise, the connection of 
other observation gadgets can be explained similarly.

\subsection{Proof of Theorem \ref{thm:CTPtheorem}}\label{sec:ThmProof}

Given a QBF $\Phi$ with $n$ variables and $m$ clauses, we construct 
$G_\Phi$ as in Section \ref{sec:CTP-cons}. Set $L=8m+16$ and 
$p_1=1-2^{-\left\lceil \log_2(\frac{3L+1}{2})\right\rceil}$. 
We show that it is optimal to traverse $(s,v_0)$ if and only if $\Phi$ is satisfiable.

Unless stated otherwise, we henceforth consider only reasonable policies for $G_\Phi$ that
do not begin with the default action of traversing $(s,t)$.
Due to properties of the gadgets (claims \ref{clm:BGcLaim}, \ref{clm:OGclaim}) any reasonable
policy $\pi$ for $G_\Phi$ must follow the restrictions in Table \ref{tab:policy}, as any other action is suboptimal. 

\begin{table*}[!h]
	\centering
	  \caption{Reasonable policy actions in $\pi$}
		\begin{tabular}{|l |r|}\hline
Location           & Action \\ \hline	
$v'_i$, for $i<n$   & cross $BG(v'_i,v_{1+1})$  \\ \hline	 
$v_i$, for $i\leq n$, & go to $v_{i1}$ or $\bar{v}_{i1}$  \\ \hline  
$v_{il}$, for $i\leq n$,   & cross $OG(v_{il},v'_{il},o_{il})$  \\ \hline  
 $\bar{v}_{il}$ for $i\leq n$, &  cross $OG(\bar{v}_{il},\bar{v}'_{il},\bar{o}_{il})$           \\ \hline	
$z_l$, for $0<l\leq m+2$   & cross $BG(z_l,z_{l-1})$  \\ \hline	
$r_0$   &   pass exam or take shortcut         \\ \hline

		\end{tabular}
	
	\label{tab:policy}
\end{table*}

Most of these restrictions are immediate consequences of executing optimal policies at the baiting and observation gadgets (see \ref{app:ApA1} and \ref{app:ApA2} for detail). The following claim, proved in \ref{app:ApB1}, shows the actions of any reasonable policy for $G_\Phi$ at $r_0$.

\begin{claim}\label{clm:r-pol}
At $r_0$, any reasonable policy acts as follows:

\begin{itemize}
\item If all the edges in the exam section were observed to be unblocked, cross $(r_0,r^1_1,\cdots,r^{m+1}_4,t)$ until reaching $t$ for a cost of $2(m+1)$. 
\item Otherwise, cross the cost $L$ shortcut edge $(r_0,t)$.
\end{itemize}
\end{claim}

\vspace{3mm}

Therefore, reasonable policies for $G_\Phi$ differ only in the choices made in the universal and existential edges,
and in the choice at $r_0$ which is either to traverse the exam section if
all clause edges were observed, or otherwise take the expensive shortcut $(r_0,t)$.

\vspace{3mm}

Now let $\pi$ be a reasonable policy for $G_\Phi$, and denote the expected 
cost of $\pi$ by $C(\pi)$. Define a \textit{weather} to be an assignment of $\{traversable,blocked\}$ to the edges of $G_\Phi$.
Let $W$ be the set of all possible weathers for $G_\Phi$, and for $w\in W$ let
$p_w$ be the probability that weather $w$ occurs. 
Define $C(\pi,w)$ to be the cost of executing $\pi$ over a 
weather $w$. Then

\begin{equation}
C(\pi)=\sum_{w\in W}p_wC(\pi,w)
\end{equation}

 Next, partition $W$ into 
\textit{full-trip weathers} $W^f(\pi)$, in which $r_0$ is reached while
executing $\pi$; and \textit{shortcut weathers} $W^s(\pi)$ in which $r_0$ is not reached due to taking a shortcut edge to $t$
before reaching $r_0$.
Then:

\begin{equation}
C(\pi)=\sum_{w\in W^s(\pi)}p_wC(\pi,w)+\sum_{w\in W^f(\pi)}p_wC(\pi,w)
\end{equation}

\noindent Let $\pi^T$ be a policy for $G_\Phi$ such that in every 
subsection $X_i$ of the variables section, whenever possible, the \textit{true-path} is always chosen. Define: 

\begin{equation}
D_{st}=\sum_{w\in W^s(\pi^T)}p_wC(\pi^T,w)
\end{equation}

\noindent As all the \textit{true-paths}, and \textit{false-paths} of all the variables section are 
symmetric in the number of observation gadgets and other edges,
there is a bijection $g_{\pi}:W^s(\pi)\to W^s(\pi^T)$ such that 
$p_w=p_{g_{\pi}(w)}$ and $C(\pi,w)=C(\pi^T,g_{\pi}(w))$
for every $w\in W^s(\pi)$.
Hence we have:

$$D_{st}=\sum_{w\in W^s(\pi)}p_wC(\pi,w)$$ 

and

\begin{equation}
C(\pi)=D_{st}+\sum_{w\in W^f(\pi)}p_wC(\pi,w)
\end{equation}

\noindent 
Again, due to symmetry, and the properties of the baiting and observation gadgets (claims \ref{clm:BGcLaim}, \ref{clm:OGclaim}),
the total cost from $s$ to $r_0$ while executing $\pi$ in 
any weather $w\in W^f(\pi)$ is independent of $w$. We denote this cost by $D_{pt}$, and can compute it simply by summing the cost
of traversing from $s$ to $r_0$ through the variable section and guard section, assuming that $r_0$ is reached. Then we get:

\begin{equation}
D_{pt}=1+(2+\frac{19mL+4}{4})n+(n+m+1)L
\end{equation}

\noindent Then from $r_0$ to $t$ the cost is either $2(m+1)$ (if the exam section is known to
be completely unblocked), 
or $L>2(m+1)$ (taking the shortcut $(r_0,t)$, if some edges in the exam section are known to be blocked, or some
such unknown edges remain).
Hence for any full-trip weather $w$, $C(\pi,w)$ is either $D_{pt}+L$, or $D_{pt}+2(m+1)$.

Let $P^\pi_\Phi \in [0,1]$ be the probability that \textbf{not} all the clause edges of the exam section 
were observed in a full-trip weather by following $\pi$
(this probability depends on the formula $\Phi$).
Then, with probability $(1-P^\pi_\Phi)(1-p_1)^{3m+2}$ all the edges of the exam section were 
observed and were found unblocked before reaching $r_0$.
Let $P_{rt}=(1-p_1)^{3m+2}$ be the probability that all the edges in the exam section are unblocked
and denote by $P_{r_0}$ the probability of
reaching $r_0$ by executing $\pi$. 
Again, due to symmetry of the baiting and observation gadgets, $P_{r_0}$ is independent of $\pi$. 
We get:

\begin{equation}
\sum_{w\in W^f(\pi)}p_wC(\pi,w)=P_{r_0}(D_{pt}+P^\pi_\Phi L+(1-P^\pi_\Phi)(P_{rt}2(m+1)+(1-P_{rt})L)
\end{equation}

And therefore

\begin{equation}
C(\pi)=D_{st}+P_{r_0}(D_{pt}+P^\pi_\Phi L+(1-P^\pi_\Phi)(P_{rt}2(m+1)+(1-P_{rt})L))
\end{equation}

If $\Phi$ is satisfiable, then, as in the proof of Theorem \ref{thm:CTPDep}, there
is a reasonable policy $\pi$ which follows the 
assignments that satisfy $\Phi$, thus every clause edge is observed, 
and $P^\pi_\Phi=0$. Define $B_0=C(\pi)$ for such a policy $\pi$ when $\Phi$ is satisfiable. Then

\begin{equation}
B_0 = D_{st}+P_{r_0}(D_{pt}+P_{rt}2(m+1)+(1-P_{rt})L)
\end{equation}

If $\Phi$ is not satisfiable, then at some universal subsection $X_i$ of the variables section, 
there is a probability of at least $1/4$ that a universal edge must be traversed, 
such that upon reaching $r_0$, not all the clause edges are visited. 
Hence, in total, there is a probability of at least 
$(\frac{1}{4})^{\frac{n}{2}}$ that not all the clause edges are visited. 
Note that as $P_{r_0}$ already excludes events where both universal edges are blocked for
some variable, if $\Phi$ is not satisfiable, then for every reasonable policy $\pi$, $P^\pi_\Phi >(\frac{1}{3})^{\frac{n}{2}}$.
Hence define $B_1$ as follows.

\begin{equation}
B_1=D_{st}+P_{r_0}(D_{pt}+(\frac{1}{3})^{\frac{n}{2}}L+(1-(\frac{1}{3})^{\frac{n}{2}})(P_{rt}2(m+1)+(1-P_{rt})L))
\end{equation}

Then $B_1>B_0$, and if $\Phi$ is not satisfiable, then $C(\pi)\geq B_1$.
Now let $h=w((s,t))=B_0+(\frac{1}{4})^{\frac{n}{2}}mP_{r_0}$,
so that $B_1 > h > B_0$. Thus the optimal action at $s$  is
to traverse $(s,t)$ if and only if $\Phi$ is unsatisfiable. 
Since the CTP-graph construction used a polynomial number of vertices,
and all the weights and probabilities by construction need only a polynomial number of bits
(see \ref{app:ApB2} for the technical computation of $h$), Theorem \ref{thm:CTPtheorem} follows.
\qed
\vspace{2mm}

Several corollaries follow due to the construction of $G_\Phi$:

\begin{corollary}
It is PSPACE-hard to determine the expected cost of the optimal policy in CTP.
\end{corollary}

By replacing all the edges with appropriately directed edges, we get:

\begin{corollary}
CTP with directed edges but no dependencies is PSPACE-complete.
\end{corollary}

Finally, as every unknown edge in this construction of $G_\Phi$ has cost $0$ 
and a probability which is a power of $2$ of being unblocked (the universal edges, for example,
can be split into a two-edge path), 
we can replace every unknown edge with a path of zero-cost, blocking probability  $1/2$
edges and get:

\begin{corollary}
CTP remains PSPACE-complete even if all the unknown edges have zero cost and blocking probability $1/2$.
\end{corollary}

\section{Complexity of CTP with remote sensing}\label{sec:sensing}

A somewhat more general version of CTP is {\em Sensing-CTP}. In
this variant, the state of
graph edges can be remotely sensed at any time, paying a known cost.
Formally, Sensing-CTP is defined exactly the same way as stochastic CTP, w.r.t.
the graph, edge-blocking probabilities and weights, and source and target
vertices (see Section \ref{sec:intro}). In addition, a sensing
cost function $SC:V\times E \rightarrow \Re^+$ is given. An edge $e$, not necessarily incident on a vertex $v$,
can be \textit{sensed} for a cost $SC(v,e)$ and as a result, the true state of $e$ is revealed.
The problem in Sensing-CTP is to find a policy which minimizes the total expected 
sensing cost plus travel cost from $s$ to $t$.

CTP is solvable in polynomial time when the
graph consists of edge-disjoint paths which meet only at $s$ and $t$ \cite{BFS}.
This gives rise to the question whether Sensing-CTP is also
tractable for disjoint paths.
We show that this is not the case for Sensing-CTP unless P=NP.
Again, since the size of the policy may be exponential in the size of the
graph, we actually show that it is NP-hard to determine even the first
action in an optimal policy.

\begin{theorem}
Sensing-CTP is NP-hard even in disjoint-path graphs.
\end{theorem}

\inlineproofenvironment\ By reduction from the NP-complete problem vertex cover (VC) \cite{GJ}.
Let $G = (V, E)$ be an undirected graph, for which we
need to decide if there is an $S\subseteq V$ of size at most $k$ 
such that every edge in $E$ is incident on a vertex in $S$ 
($S$ is called a vertex cover of size $k$). 
The idea of the proof is, informally, for a policy to benefit if it 'has sensed'
all of the edges in a given VC instance, where sensing of all neighbors 
of a vertex can be done at some constant cost. By tweaking constants,
these actions will be beneficial if and only if it is possible to sense all edges in
the VC instance by sensing the neighbors of at most $k$ vertices. 

Construct the corresponding
CTP-PATH graph $G^{\prime}=(V^{\prime}, E^{\prime})$ as follows 
(see Fig. \ref{fig:VCreduction}).
For each vertex $v \in V$, construct a ``vertex'' node 
$f(v) \in V^{\prime}$, an edge $(s,f(v))$ with a cost $C$ defined below, 
and an infinite-cost edge $(f(v),t)$.
Also construct the ``default edge'' $(s,t) $ with cost $4$.
The above edges are always traversable.

Construct the ``sensing path'', a path consisting of a

``leader'' always traversable edge $e_0$ with cost $L$
starting at $s$, and for each edge $e\in E$, a zero-cost 
edge $f(e)$ in sequence, and finally an infinite-cost edge $(u, t)$.
The probability that each of the latter edges is
blocked is $\epsilon $, a small positive number defined below, except for $(u,t)$ which has infinite-cost,
therefore is never traversable. Additionally, we have a two edge ``uncertain'' path
$((s, x), (x, t))$ where $(s,x)$ is always traversable and costs $2$, and
$(x, t)$ is traversable with probability $1/2$ and costs $0$.
Note that the resulting graph $G^{\prime}$ consists only of edge-disjoint
paths leading from $s$ to $t$.

\begin{figure}[htb]
  \centering
\scalebox{1}{\includegraphics{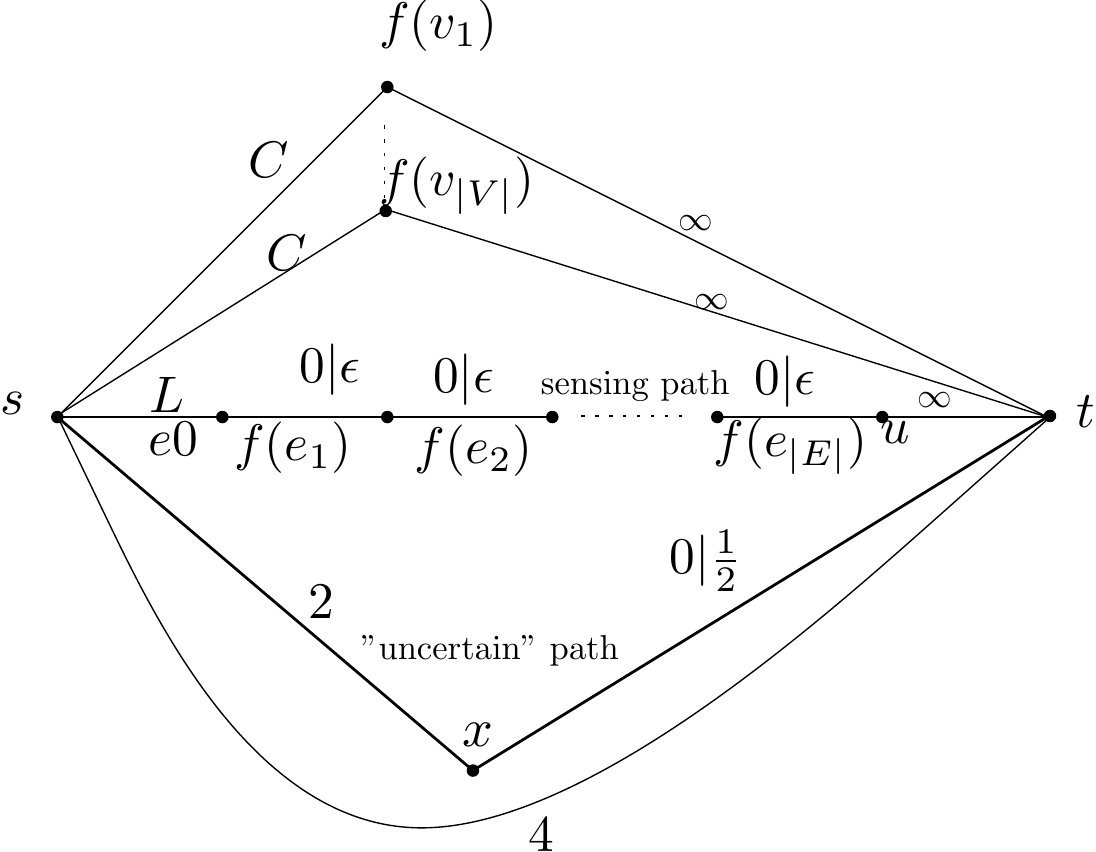}}

\vspace{-0.3cm}
\caption{CTP-graph for reduction from vertex cover}
    \label{fig:VCreduction}\vspace{-0.3cm}
\end{figure}

The sensing cost from ``vertex'' node $f(v)$ on edge $f(e)$ is 0 for all $e$ incident on
$v$ in $G$. Sensing $(x,t)$ costs $0$ from $u$.
All other sensing costs are infinite. We show that for parameter values
defined below, the optimal policy is to immediately traverse $(s, t)$ if and only if
$G$ does not have a vertex cover of size $k$. 

Assuming no sensing, there are only two ``reasonable'' traversal policies:
traversing $(s,t)$ immediately for a cost of $4$, or trying the uncertain path, which costs
$2$ if $(x,t)$ is traversable, and 8 if not, for an expected cost of $5$.
Thus the optimal policy with no sensing is to traverse $(s,t)$. However,
if sensing is allowed, the ``value of information'' of knowing
whether $(x,t)$ is traversable is: $VOI(x,t) = 4-(2\times 1/2+4\times 1/2)=1$,
due to the fact that if $(x,t)$ is revealed as traversable (which happens with
probability $1/2$) we gain 2 by taking the "`uncertain"' path instead of $(s,t)$, and otherwise gain nothing.

We show below that for appropriate values of $L$, it is not beneficial to try to get to $u$ in order
to sense $(x, t)$ unless sensing reveals that the path to $u$ is unblocked.
The probability that at least one edge on the sensing path is blocked
is at least $\epsilon $.
However, all the edges in the sensing path can be sensed (stopping if any 
blocked edge is sensed) for a cost $C_{sense}$ defined below.
If $G$ has a vertex cover of size $k$, then  $C_{sense}\leq 2Ck$
by visiting $k$ ``vertex'' nodes.
However, if the smallest vertex cover is of size $k'\geq k+1$, 
the expected sensing cost becomes:

\begin{equation}
C_{sense} =  2C\sum_{i=1}^{k'} \prod_{j=1}^{i-1} (1-\epsilon )^{d_j}\geq 2Ck'(1-\epsilon )^{|E|}
\end{equation}

\noindent where $d_j$ is the number of previously unsensed edges of the $j$th vertex
in the (unknown) optimal sensing order. For any $0<\alpha<1$, set $\epsilon$ such that 

\[
\epsilon =  1-(\frac{k+1-\alpha}{k+1})^{\frac{1}{|E|}}
\]

Then, as $2Ck'(1-\epsilon )^{|E|}$$ \geq 2C(k+1-\alpha)$, we have that

$$C_{sense} \geq 2C(k+1-\alpha)$$

Now set $L= \frac{1}{2}-\frac{\epsilon}{4}$ and:
\[
C = \frac{\epsilon(1-\epsilon)^{|E|}}{2(2k+1-\alpha)}
\]

To complete the proof, it is sufficient to prove the following claim:

\begin{claim}
The optimal action at $s$ is to traverse $(s, t)$ if and only if
$G$ does not have a vertex cover of size $k$. 
\end{claim}

\begin{proof} We show the following: 

\begin{itemize}

\item If there is no cover in size $k$, the optimal policy is to traverse (s,t). 

\item Otherwise, the optimal policy is as follows: \textit{visit the ``vertex'' nodes
constituting the cover, doing the appropriate sensing actions; if the
sensing path is unblocked, visit $u$ in order to sense $(x,t)$, and
then take the path from $s$ to $t$ through $x$ if $(x,t)$ is found traversable.} (In this case, given the optimal policy, one can straightforwardly 
construct the vertex cover of size $k$.)
\end{itemize}

Note that it is suboptimal
to try the sensing path unless assured that all
edges leading to $u$ are traversable, because then there is
at least one edge that can be blocked with probability $\epsilon $,
in which case attempting this path results in no positive gain.
To see this, note that the traversal cost $2L$ must be paid anyway, 
thus the total expected gain from trying the sensing path is:

\[
g \leq ( 1-\epsilon) VOI(x,t) -2L = 1-\epsilon -2L = 
1-\epsilon -2(\frac{1}{2}-\frac{\epsilon}{4}) = -\frac{\epsilon}{2} < 0
\]

This also holds for all policies that attempt some sensing actions before
trying the sensing path, but that do not make sure that $u$ is reachable
before trying the sensing path.

Now, if there is a vertex cover of size $k$, the expected cost of
sensing all edges in the sensing path is at most $2Ck$. 
If $u$ is found to be reachable (prior probability $(1-\epsilon)^{|E|}$ of that happening)
use the sensing path (which costs $2L$)
to sense $(x,t)$, gaining the expected $VOI(x,t)$ of 1.
The total expected gain in this case is positive:
\[g'\geq
(1-\epsilon)^{|E|} (VOI(x,t)-2L)-2Ck = 
(1-\epsilon)^{|E|}(1-(1-\frac{\epsilon}{2})) - 
2k\frac{\epsilon(1-\epsilon)^{|E|}}{2(2k+1-\alpha)}
\] 
\[
= \epsilon(1-\epsilon)^{|E|}(\frac{1}{2}-\frac{2k}{2(2k+1-\alpha)}) > 0
\]
 
Finally, if there is no vertex cover of size $k$, then the expected gain for
sensing in the best case is (assuming the policy of sensing all edges in the
sensing path from only $k+1$ ``vertex'' nodes, costing
at least $2C(k+1-\alpha )$ in expectation) is at most:
\[g''\leq
(1-\epsilon)^{|E|}(VOI(x,t)-2L) - 2C(k+1-\alpha ) = 
\]
\[
(1-\epsilon)^{|E|}(1-(1-\frac{\epsilon}{2})) - 
2(k+1-\alpha)\frac{\epsilon(1-\epsilon)^{|E|}}{2(2k+1-\alpha)}
\]
\[
 = \epsilon(1-\epsilon)^{|E|}(\frac{1}{2}-\frac{2(k+1-\alpha)}{2(2k+1-\alpha)}) < 0
\]
and thus the optimal policy here is not to perform sensing at all, but to
traverse $(s,t)$ immediately. \qed

\end{proof}

\section{Discussion}

Having shown that stochastic CTP is PSPACE-hard, several related questions on
variants of CTP and CTP with restricted topologies arise.
One issue of particular interest is the question of efficiently finding
approximately optimal actions. The proofs in this paper make use of rather small gaps between expected
values of two candidate actions, and thus leave open the possibility of efficient
approximation algorithms.

Studies of the \textit{competitive analysis} of the Canadian Traveler
Problem reveal rudimentary bounds on approximability. Denoting by $k$
the number of uncertain edges in an instance, there exists for the
undirected case polynomial-time algorithms achieving competitive
ratios of $2k+1$ \cite{SW}. As a consequence, stochastic CTP can be
approximated within $2k+1$. With a slightly improved analysis, the
same algorithm yields a $2n+1$-approximation. In the directed case,
existing results from competitive analysis only yield approximations
of $2^{k+1}+1$ and $2^{n+1}+1$, respectively \cite{XHSZZ}.

These approximation algorithms forego entirely the stochastic nature of the problem
and leave open considerable improvements. At the time of this writing, no notable
hardness of approximation results are known.

Another issue is: what is the most general graph topology under which CTP is tractable
or easy to approximate? An efficient algorithm was shown \cite{BFS} for disjoint-path
graphs, based on a lemma that there exists an optimal policy that is committing:
(such a policy never returns to the source vertex unless its 'current' path to the target vertex is known
to be blocked). But any departure from the disjoint-path structure (such as adding even one more
edge that crosses between vertices in two of the paths) complicates things considerably
by voiding the optimality of committing policies. 

The \textit{st}-reliability problem (finding the probability that an unblocked path from $s$ to $t$ exists)
appears to be an essential building block in solving CTP. The reduction in \cite{PY} shows that CTP is 
intractable for almost any graph topology for which \textit{st}-reliability is intractable as well.
An open research question is whether CTP is tractable or easy to approximate for
graph topologies in which \textit{st}-reliability is tractable, such as ``tree-structured'' graphs,
or the more general series-parallel graphs. 

Sensing-CTP is generally harder than CTP for restricted topologies,
as shown in Section \ref{sec:sensing}.
Other variants of CTP, such as CTP-Dep, are also harder than CTP, as
dependencies can act like remote sensing. It was shown that CTP-Dep is NP-hard for disjoint
path graphs \cite{zarchy10thesis}. However, when considering topological restrictions, one must
also consider the topology of the dependency-graph, and the hardness proof in \cite{zarchy10thesis}
used an essentially unrestricted topology Bayes network to represent the dependencies.

\section{Acknowledgments}
This research is partially supported by the Israel Science
Foundation grant 305/09, the Lynn and William Frankel
Center for Computer Sciences, and by ERC Advanced
Investigator Grant 226203.

\bibliographystyle{elsarticle-num}
\bibliography{DPhD-cites}

\appendix
\section{}\label{app:ApA}

\subsection{Baiting Gadgets}\label{app:ApA1}

\vspace{3mm}

Let $g=BG(u,v)$ be a baiting gadget with a parameter $L>1$,  defined in Section \ref{sec:BG} (see Fig.\ \ref{fig:BGsubfig1}). 
Recall that $\pi$ (as defined in Section \ref{sec:BG}) is the following policy for $g$:
\textit{At $u$, traverse $(u,v_1)$. At $v_i$, for any $i<N$, do as follows: if $(v_i,t)$ is unblocked, reach the destination through $(v_i,t)$ for cost of zero. However, if $(v_i,t)$ is blocked, traverse $(v_i,v_{i+1})$ for a cost of $L/(N+1)$. At $v_N$, if $(v_N,t)$ is blocked, traverse $(v_N,v)$}. 

By construction of the CTP-graph, we assume that any policy other than traversing $(u,v)$ results in a cost of at least $1$ (see Section \ref{sec:BG}).
\vspace{2mm}

Apart from $\pi$, other policies at $u$ that are not clearly suboptimal are:

\begin{itemize}
	\item \textit{choose not to traverse $(u,v_1)$}.
	
	\item The following type of policies denoted by $\pi_j$, for $j\leq N$ : \textit{execute $\pi$ until reaching $v_j$; if $(v_j,t_j)$ is unblocked, reach the destination through $(v_j,t_j)$; otherwise, retreat to $u$ and execute an optimal policy with an expected cost of $M_j\geq 1$.}
	
\end{itemize}

Finally, we set $N=2^{\left\lceil log_2(4L)\right\rceil}-1$, implying $N+1\geq 4L$.

\setcounter{claim}{1}

\begin{claim}
When at $u$ for the first time, $\pi$ is optimal for a baiting gadget $BG(u,v)$ with a parameter $L>1$.
After reaching $v$, it is suboptimal to backtrack to $u$.
\end{claim}

\begin{proof}
Denote by $K$ the expected cost of the optimal policy executed once $v$ is reached. As there is an $L$ cost shortcut edge $(v,t)$, it is clear that $K\leq L$ therefore it is always suboptimal to retrace $g$ once $v$ is reached. 
We first show that $C(\pi)<1$, hence choosing not to traverse $(u,v_1)$ is suboptimal. 
 
Note that for every $i\leq N$, the probability that $(v_i,v_{i+1})$ is traversed in $\pi$ is $(\frac{1}{2})^i$. Hence we have

\begin{equation}
C(\pi)=\frac{L}{N+1}\sum_{i=0}^N(\frac{1}{2})^i+(\frac{1}{2})^NK
\end{equation}

\noindent Thus

\begin{equation}\label{eq:AppAeq1}
C(\pi)=\frac{2L}{N+1}(1-2^{-(N+1)})+2^{-N}K
\end{equation}

Then, as $K\leq L$, $N+1\geq 4L$ and $L>1$, we have that

\begin{equation}
C(\pi)<\frac{2L}{4L}+2^{-N}L<\frac{3}{4}<1.
\end{equation}

\noindent As required.

\vspace{4mm}

Next we show that for every $j\leq N$, $C(\pi)<C(\pi_j)$, hence for every $j\leq N$, the policy $\pi_j$ is suboptimal. We have that

\begin{equation}
C(\pi_j)=\frac{L}{N+1}\sum_{i=0}^{j-1}(\frac{1}{2})^i+(\frac{1}{2})^j(\frac{jL}{N+1}+M_j)
\end{equation}

\noindent Thus

\begin{equation}
C(\pi_j)=\frac{2L}{N+1}(1-2^{-j})+\frac{2^{-j}jL}{N+1}+2^{-j}M_j
\end{equation}

As $K\leq L$, and $1< M_j$, it is sufficient from (\ref{eq:AppAeq1}) to show that for every $0<j\leq N$

\begin{equation}
\frac{2L}{N+1}(1-2^{-(N+1)})+2^{-N}L < \frac{2L}{N+1}(1-2^{-j})+\frac{2^{-j}jL}{N+1}+2^{-j}
\end{equation}

Therefore

\begin{equation*}
\frac{2L}{N+1}(2^{-j}-2^{-(N+1)}-2^{-j-1}j)+2^{-N}L < 2^{-j}
\end{equation*}

Hence we need to show that for every $0<j\leq N$,

\begin{equation*}
\frac{2L}{N+1}(1-2^{j-N-1}-\frac{j}{2})+2^{j-N}L <1
\end{equation*}

As $N+1\geq 4L$ and $L>1$, it is sufficient to show that for every $0<j\leq N$, 

\begin{equation}\label{eq:AppAeq2}
\frac{1}{4}(2-j)+2^{j-N}L<1
\end{equation}

And inequality (\ref{eq:AppAeq2}) follows since the function

\begin{equation*}
f(x)=\frac{1}{4}(2-x)+2^{x-N}L
\end{equation*}

over the reals, has only one extremum, $f(0)<1$ ,$f(N)<1$  and 

$$\lim_{x\to\infty}f(x)=\lim_{x\to-\infty}f(x)=\infty$$

\end{proof}

\subsection{Observation Gadgets}\label{app:ApA2}

\vspace{3mm}

Let $g=OG(u,v,o)$ be an observation gadget as defined in Section \ref{sec:OG}, and seen in Fig.\ \ref{fig:OGsubfig2}.
Recall that $L>8$, and $L_1=5L/8$. $\pi_g$ is the following partial policy for $OG(u,v,o)$: \textit{At $u$, cross $BG_1$ (observe $(v_1,v_4)$). Then cross $BG_2$. If either $(v_1,v_4)$, or $(v_2,v_3)$ is found blocked, reach $t$ by traversing the $3L/2$ cost shortcut edge $(v_2,t)$. However, if both $(v_1,v_4)$, and $(v_2,v_3)$ are unblocked, traverse $(v_2,v_3,o,v_4,v_1,v'_1)$ (at $o$, observe the edges incident on $o$, in case there are any), and cross $BG_3$. }
 
We again assume, by construction of the CTP-graph,
that any policy at $u$ other than crossing $BG_1$ results in a cost of at least $1$.

\begin{claim}
When at $u$ for the first time, $\pi_g$ is an optimal policy for $g$.
\end{claim}

\begin{proof}

At $u$, as $BG_1$ is a baiting gadget, then by Claim \ref{clm:BGcLaim}, it is optimal to cross $BG_1$. 
When first arriving at $v_1$, after $BG_1$ is crossed, $(v_1,v_4)$ is observed. As $(o,v_4)$ has a cost of $L_1>1$, and $(v_1,v'_1)$ has a cost of $1$, then by Claim \ref{clm:BGcLaim}, it is optimal to cross $BG_2$. Once at $v_2$, if $(v_2,v_3)$ is blocked, it is optimal to take the shortcut $(v_2,t)$ for a cost of $3L/2$.

It remains to show that if $(v_2,v_3)$ is unblocked, the optimal policy at $v_2$ is:

\begin{enumerate}

\item if $(v_1,v_4)$ is unblocked, traverse $(v_2,v_3,o,v_4,v_1,v'_1)$, and cross $BG_3$.

\item otherwise, traverse the shortcut $(v_2,t)$ for a cost of $3L/2$.

\end{enumerate}

\noindent\textbf{case 1: } $(v_1,v_4)$ is unblocked. 

First note that arriving at $v_1$ a second time through $(v_4,v_1)$, $BG_1$ and $BG_2$ are known not to have any blocked shortcut edges, hence by Claim \ref{clm:BGcLaim}, the optimal policy when arriving at $v_1$ a second time is to traverse ($v_1,v'_1)$ and the baiting gadget $BG_3$.

Now, traversing $(v_2,v_3,o,v_4,v_1,v'_1)$, and crossing $BG_3$ bears an expected cost of at most $2L_1+2$, while traversing $(v_2,t)$ costs $3L/2$. 
Hence, as $2L_1+2< 3L/2$, it is optimal at $v_2$ to traverse $(v_2,v_3,o)$. 

We now inspect the possible partial policies at $o$:

\begin{description}
	\item[case 1.a:] Traverse $(o,v_4,v_1,v'_1)$, and cross $BG_3$ for an expected cost of at most $L_1+2$. We denote this partial policy by $\pi'$.
	
	\item [case 1.b: ] Traverse edges of another observation gadget $\widetilde{g}$, in case there exists such $\widetilde{g}$ incident on $o$. 
Suppose that $\widetilde{g}$ has not already been traversed (label the vertices of $\widetilde{g}$ as $\widetilde{v}_i$). Then traversing either $(o,\widetilde{v}_3)$ and trying to traverse $(\widetilde{v}_3,\widetilde{v}_2)$, or traversing $(o,\widetilde{v}_4)$ and trying to traverse $(\widetilde{v}_4,\widetilde{v}_1)$ results in an expected cost of at least $L_1+3L_1/4$. Hence, as $L_1+2<L_1+3L_1/4$, we have that executing $\pi'$ is cheaper than traversing any edges of $\widetilde{g}$.

Next suppose that $\widetilde{g}$ has already been traversed. Therefore we may assume that the policy $\pi_{\widetilde{g}}$ was executed in $\widetilde{u}$, thus the baiting gadgets of $\widetilde{g}$ are known not to contain any unblocked zero-cost shortcuts, hence crossing each such baiting gadget costs $L$. Then traversing $\widetilde{g}$ results in an expected cost of at least $L_1+L$, and as $L_1+2<L_1+L$, we again have that executing $\pi'$ is cheaper than traversing any edges of $\widetilde{g}$.

\item [case 1.c:] Traverse the exam section path. 
Recall that $o$ is identified with $r_5$. Suppose that observation edge $(r_4,r_5)$ is blocked. At $o$, denote the following partial policy by $\pi_1$: \textit{cross $(r_5,r'_1)$; if $(r'_1,r'_2)$ is unblocked, continue with any optimal policy, otherwise, return to $o$, and execute $\pi'$, which can still be executed, for an expected cost of $C(\pi') < L_1+2$}.
Then we have 

$$C(\pi_1)\geq 1+p_1(1+C(\pi'))$$

and as $p_1>1-2/(3L+1)$, we have that 

$$C(\pi')<1+p_1(1+C(\pi'))< C(\pi_1)$$

Therefore executing $\pi'$ is cheaper than executing $\pi_1$.

Now suppose that $(r_4,r_5)$ is  unblocked. Then we can either execute $\pi_1$ (or the symmetric case in which $(r_2,r_3)$ is being inspected) with the same analysis, or we can extend $\pi_1$ with the following policy denoted by $\pi_2$:

\textit{execute $\pi_1$; upon returning to $o$ (after $(r'_1,r'_2)$ is found blocked), cross $(r_5,r_4)$ and $(r_4,r_3)$; if $(r_3,r_2)$ is unblocked, continue with any optimal policy; otherwise return to $r_5$, and execute $\pi'$, which can still be executed for an expected cost of $C(\pi') < L_1+2$}. Then we have that

$$C(\pi'_2)\geq 1+2p_1+p_1^2(1+C(\pi'))$$ 

However, $p_1>1-2/(3L+1)$ entails $C(\pi')<C(\pi'_2)$. Therefore executing $\pi'$ is cheaper than executing $\pi_2$
as well. The policy in which $(r_2,r_3)$ is the first edge among $(r_2,r_3)$ and $(r'_1,r'_2)$ to be inspected is symmetric to $\pi_2$.
Hence we see that traversing any edges of the exam section path is suboptimal.

\end{description}

\noindent \textbf{case 2:} $(v_1,v_4)$ is blocked. 

In this case the following partial policies can be executed at $v_2$:
\begin{description}
\item[case 2.a:] Take the shortcut edge $(v_2,t)$ for a cost of $3L/2$. Denote this policy by $\pi'$.

\item[case 2.b:] Traverse $(v_2,v_3)$ and $(v_3,o)$ for a cost of $L_1$ and at $o$ traverse edges of another observation gadget. As in case 1.b we have 
that traversing edges of $\widetilde{g}$ results in an expected cost of at least $L_1+3L_1/4$. Then, as $3L/2<L_1+L_1+3L_1/4$, we have that $\pi'$ is cheaper than reaching $o$ and traversing any edges of $\widetilde{g}$.

\item[case 2.c:]Traverse $(v_2,v_3)$ and $(v_3,o)$, for cost of $L_1$, and at $o$ traverse the exam path. We define $\pi_1$ and $\pi_2$ as in case 1.c.
Recall that $C(\pi')=3L/2$. However as $p_1>1-2/(3L+1)$, we have that $C(\pi')<L_1+C(\pi_1)$, and $C(\pi')<L_1+C(\pi'_2)$. 
Hence $\pi'$  is cheaper than traversing to $o$ and traversing the any edges of the exam section path.
\end{description}

\qed
	
\end{proof}

\section{}\label{app:ApB}

\subsection{Behavior of reasonable policies}\label{app:ApB1}

\begin{claim}
At $r_0$, any reasonable policy acts as follows. If all the edges in the exam section were observed to be unblocked, cross 

\noindent $(r_0,r^1_1,\cdots,r^{m+1}_4,t)$ until reaching $t$ for a cost of $2(m+1)$. Otherwise, cross the cost $L$ shortcut edge $(r_0,t)$.
\end{claim}

\begin{proof}

Retracing $BG(z_0,z_1)$ clearly results in a cost of at least $L$. We first see that unless all the edges in the exam section were observed to be unblocked, any partial policy executed at $r_0$ results in a cost of at least $L$, therefore it is cheaper to take the shortcut $(r_0,t)$ for a cost of $L$. 

At every vertex $r^l_i$, $l\leq m+1$, $,i\leq 5$, any unblocked edge on the exam section path, incident on $r^l_i$, can be traversed. 
At $r^l_2$ there is an additional option to cross either $BG(z_l,z_{l+1})$, or $BG(z_{l-1},z_l)$ which hold no unblocked shortcut edges, hence crossing these results in a cost of at least $L$. If $r^l_5$ is identified with an observation point of some observation gadget $g'$, there is an additional option to traverse edges of $g$. However by an argument identical to case 1.b of \ref{app:ApA2}, traversing any edges of $g$ results in a cost of at least $L$. Hence any deviation from the exam section path results in a cost of at least $L$. 

Suppose that all the edges of the exam section are known to be unblocked. Then, as the exam section contains $2(m+1)$ always traversable cost $1$ edges, and as  $2(m+1)<L$, the optimal policy is to cross the exam section  $(r_0,r^1_1,\cdots,t)$ for the cost of $2(m+1)$.

Otherwise, suppose there are edges in the exam section with unknown status. As all the guard edges were observed upon crossing the guard section,then such an edge is a clause edge. Hence let $e_l=(r^l_4,r^l_5)$ be the first unknown clause edge such that every edge in the path $(r_0,\cdots,r^l_4)$ is known to be unblocked. Finding $(r^l_4,r^l_5)$ blocked results in either retracing the exam section to $r_0$ and taking the cost $L$ shortcut to $t$, or in deviating from the exam section. Hence, as $(r^l_3,r^l_4)$ costs $1$, $L>1$ and $p_1>1-2/(3L+1)$, traversing from $r_0$ to $e_l$ results in an expected cost of at least $1+p_1(1+L)>L$. Hence traversing the shortcut edge $(v_2,t)$ is cheaper.
Obviously, the same argument holds for traversing $e_l$ where $e_l$ is previously known to be blocked.
\qed

\end{proof}

\setcounter{claim}{5}

\subsection{Polynomial size representation}\label{app:ApB2}

We show the computation of $h$, the cost of the default edge $(s,t)$. 
Recall that $L=8m+16$, 
$N=2^{\left\lceil \log_2(4L)\right\rceil}-1$
and
$p_1=1-2^{-\left\lceil \log_2(\frac{3L+1}{2})\right\rceil}$.

From Section \ref{sec:ThmProof} we have that

$$h=B_0+(\frac{1}{4})^{\frac{n}{2}}mP_{r_0}$$

and

$$B_0 = D_{st}+P_{r_0}(D_{pt}+P_{rt}2(m+1)+(1-P_{rt})L)$$

where for a reasonable policy $\pi$, $P_{r_0}$ is the probability of reaching $r_0$ by executing $\pi$, $D_{pt}$ is the total cost from $s$ to $r_0$ while executing $\pi$ in a full-trip weather, $D_{st}$ is the expected cost of executing $\pi$ over shortcut weathers, and $P_{rt}$ is the probability that all the edges in the exam section are unblocked.

From section \ref{sec:ThmProof}, $P_{rt}=(1-p_1)^{3m+2}$ and 

$$D_{pt}=1+(2+\frac{19mL+4}{4})n+(n+m+1)L$$ 

Hence it is left to compute $P_{r_0}$ and $D_{st}$. To do that we define the following.

\vspace{3mm}

For $k>0$, let $G(k)$ be a CTP instance composed of a series of gadgets $g_i$, $1\leq i\leq k$, such that for every $i<k$, $Exit(g_i)$ is identified with $Entry(g_{i+1})$. The gadgets in $G(k)$ are either all baiting gadget with a parameter $L>1$ (then $G(k)$ is denoted by $BG(L,k)$), or are all observation gadgets (then $G(k)$ is denoted by $OG(k)$). Set $s$ to be $Entry(g_1)$. Then we denote the following policy for $G(k)$ as $\pi_k$:
\textit{For every $i\leq k$, cross $g_i$.} 

Let $q(G(k))$ be the probability that $Exit(g_k)$ is reached by executing $\pi_k$. Let $w_1(G(k))$ be the expected traversal cost when $Exit(g_k)$ is reached while executing $\pi_k$. Let $w_2(G(k))$ be the expected cost in case a shortcut to $t$ is taken while executing $\pi_k$. 
Then we have for $k>1$

\begin{equation}
w_2(G(k))=w_2(G(1))+q(G(1))(w_1(G(1))+w_2(G(k-1)))
\end{equation}

\vspace{3mm}

Next, if $G(k)$ is a series of baiting gadgets we have that $q(BG(L,k))=2^{-kN}$ and $w_1(BG(L,k))=kL$. From (\ref{eq:AppAeq1}) we obtain:

\begin{equation}
w_2(BG(L,1))=\frac{2L}{N+1}(1-2^{-(N+1)})-2^{-N}L
\end{equation}

If $G(k)$ is a series of observation gadgets, set

$$N_1=2^{\left\lceil \log_2(\frac{3}{2}4L)\right\rceil}-1$$

Then we have $q(OG(k))=2^{-k(2N+N_1+4)}$, and $w_1(OG(k))=\frac{k(19L+4)}{4}$. We compute $w_2(OG(1))$ to be:

\begin{multline}
w_2(OG(1))=w_2(BG(L,1))+2^{-N}(w_1(BG(L,1))\\+2^{-N}(w_2(BG(\frac{3L}{2},1))+2^{-N-N_1}(\frac{3L}{2})
+2^{-N-N_1-4}(2L_1+1+w_2(BG(L,1)))))
\end{multline}

\vspace{3mm}

We can now compute $P_{r_0}$. Assume w.l.o.g. $n$ is even. Due to symmetry of the \textit{true/false-paths}, and the subsections of the variable sections, we have that

\begin{equation}
P_{r_0}= \left(q(BG(L,1))q(OG(L,m))\frac{3}{4}q(BG(L,1)) q(OG(L,m))\right)^\frac{n}{2}q(BG(L,m+2))
\end{equation}

\vspace{3mm}

To find $D_{st}$ we note the following. Assume w.l.o.g $n$ is even. Due to symmetric considerations, there are parameters $q_{st}$,  $w_{st}$ and $z_{st}$, independent of $i$, such that by executing $\pi$ at $v'_i$, $v'_{i+2}$ is reached with probability $q_{st}$ and an expected cost of $w_{st}$, and $v'_{i+2}$ is \textit{not} reached (i.e. a shortcut it taken) with an expected cost of $z_{st}$. Assume w.l.o.g. the subsection $X_i$ is universal. Then

\begin{equation}
q_{st}=q(BG(L,1))\frac{3}{4}q(OG(L,m))q(BG(L,1))q(OG(L,m))
\end{equation}

\noindent and 

\begin{equation}
w_{st}=2L+4+2mw_1(OG(L,1))
\end{equation}

\noindent We now compute

\begin{multline}
z_{st}=w_2(BG(L,1))+2^{-N}(w_1(BG(L,1))+\frac{1}{4}L+\frac{3}{4}(1+w_2(OG(L,m))+\\
q(OG(L,m))(w_1(OG(L,m))+1+w_2(BG(L,1))+\\2^{-N}(w_1(BG(L,1))+1+w_2(OG(L,m))+q(OG(L,m))+1))))
\end{multline}

Now define $\bar{z}_{st}=z_{st}+q_{st}w_{st}$, set $D_{st}^1=\bar{z}_{st}$, and for every $k> 1$ set

\begin{equation}
D_{st}^k=D_{st}^1+q_{st}D_{st}^{k-1}
\end{equation}

Then

\begin{equation}
D_{st}=D_{st}^{\frac{n}{2}}+(q_{st})^\frac{n}{2}w_2(BG(L,m+1))
\end{equation}

which concludes the computation of $h$. \qed

\end{document}